\documentclass[11pt]{article}

\title{Permutation-based simultaneous confidence bounds for the false discovery proportion}

\usepackage{amsmath}
\usepackage{caption}
\usepackage{enumerate}
\usepackage{color}
\usepackage{bbm}
\usepackage{amsfonts}
\usepackage{comment}

\usepackage{changepage}
\usepackage{float}   
\restylefloat{table}  

\usepackage{tikz}
\usepackage{pgfplots}

\usepackage{natbib}

\usepackage{bbm}
\usepackage{amsthm}
\usepackage{authblk}
\usepackage[colorlinks=true,citecolor=blue,runcolor=blue,linkcolor=blue, pdfborder={0 0 0}]{hyperref}
\usepackage[nottoc]{tocbibind}
\usepackage{graphicx}
\usepackage{subfig}
\usepackage{abstract}

\usepackage{bm}

\definecolor{darkblue}{rgb}{0.0, 0.0, 0.55}

\newtheorem{theorem}{Theorem}   

\newtheorem{proposition}[theorem]{Proposition}

\newtheorem{definition}[theorem]{Definition}
\newtheorem{assumption}[theorem]{Assumption}

\newcommand{\I}{\mathcal{I}}
\newcommand{\K}{\mathcal{K}}
\newcommand{\J}{\mathcal{J}}
\newcommand{\R}{\mathcal{R}}
\newcommand{\N}{\mathcal{N}}

\newcommand{\U}{\mathcal{U}}

\newcommand{\B}{\mathcal{B}}
\newcommand{\co}{\mathbb{T}}
\newcommand{\rset}{\mathbb{K}}

\author[1]{Jesse Hemerik\thanks{Address for correspondance: Jesse Hemerik, Department of Medical Statistics and Bioinformatics,
Leiden University Medical Center, Postzone S5-P,
Postbus 9600,
2300 RC Leiden,
The Netherlands. E-mail: jesse.hemerik@medisin.uio.no }}
\author[2]{Aldo Solari}
\author[3]{Jelle J.  Goeman}
\affil[1,3]{
Leiden University Medical Center, The Netherlands}
\affil[2]{University of Milano-Bicocca, Italy}

\begin{document}

\maketitle

\begin{abstract}
When multiple hypotheses are tested, interest is often in ensuring that the proportion of false discoveries (FDP) is small with high confidence. 
In this paper, confidence upper bounds for the FDP are constructed, which are simultaneous over all rejection cut-offs. 
In particular this allows the user to select a set of hypotheses post hoc such that  the FDP lies below some constant with high confidence.
Our method uses permutations to account for the dependence structure in the data. 
So far only Meinshausen provided an exact, permutation-based and computationally feasible method for simultaneous FDP bounds. 
We provide an exact method, which uniformly improves this procedure.
Further, we provide a generalization of this method. 
It lets the user select the shape of the simultaneous confidence bounds.
This gives the user more freedom in determining the power properties of the method.
Interestingly, several existing permutation methods, such as Significance Analysis of Microarrays (SAM) and Westfall and Young's max\emph{T} method, are obtained as special cases. 
\end{abstract}

\section{Introduction}
\label{sec:intro}

The goal of many multiple testing methods is to reject as many hypotheses as possible while incurring few type-I errors. The resulting proportion of type-I errors among the rejections is called the False Discovery Proportion (FDP). 
The FDP has received much attention in recent years since under strong dependence among the \emph{p}-values, it represents a more relevant quantity than the false discovery rate \citep{benjamini1995controlling}, the expected value of the FDP \citep{schwartzman2011effect, schwartzman2012comment,  guo2014further}.
Under strong dependence, the false discovery rate can be far from the true FDP.

In practical applications, when rejecting all hypotheses with \emph{p}-values less than an certain threshold, one would like to know a $(1-\alpha)100\%$-confidence upper bound for the FDP. The goal of this paper is to provide confidence bounds for the FDP which are simultaneous over multiple thresholds. This allows the user to freely  select the threshold \emph{post hoc}, i.e. after looking at the data, and still obtain a valid confidence bound.

There exist several methods that provide `exceedance control', i.e. control of the probability that the FDP exceeds a prespecified constant \citep{van2004multiple, farcomeni2009generalized, lehmann2012generalizations, guo2014further}. 
The number of methods allowing post hoc selection however is limited \citep{blanchard2017post}. Most of these methods (including those in the present paper) are special cases or shortcuts for the general methods in \citet{genovese2006exceedance} and \citet{goeman2011multiple}.  
The parametric methods among these (i.e. methods which rely on distributional assumptions rather than permutations to derive the null distribution) are conservative for many dependence structures of the \emph{p}-values \citep{goeman2014multiple}.

In multiple testing, when a permutation method can be used, this often offers an improvement in power over parametric procedures. The reason is that permutation methods take into account the a priori unknown dependence structure of the \emph{p}-values \citep{westfall1993resampling, meinshausen2011asymptotic,hemerik2018false}.  Parametric methods do not and are, as a consequence, often conservative. For example, under strong positive correlations among the \emph{p}-values, the Bonferroni-Holm method \citep{holm1979simple} is very conservative and its power is greatly improved by a permutation method  \citep{westfall1993resampling}.   Permutation methods are exact in the sense that the $\alpha$ level is exhausted if all hypotheses are true, and the error rate is at most $\alpha$ otherwise.
Existing permutation methods for FDP confidence are \citet{korn2004controlling, korn2007investigation}, \citet{meinshausen2005lower} and \citet{hemerik2018false}, but only \citet{meinshausen2006false} provides simultaneous FDP bounds and hence post hoc selection.
It is also the only permutation method that provides exceedance control of the FDP.
\citeauthor{meinshausen2006false}'s procedure often outperforms parametric methods.

In the present paper, the method by  \citet{meinshausen2006false} is generalized and improved. Interestingly, several well-known permutation methods are special cases of the generalization, for example the max\emph{T} method by \citet{westfall1993resampling}  and the method in \citet{hemerik2018false} \citep[an extension of Significance Analysis of Microarrays by][]{tusher2001significance}.

We improve the method in \citet{meinshausen2006false} in the following ways.
First, its power is uniformly improved by an iterative method, without additional assumptions. 
Second, as \citet{blanchard2017post} note, there is a ``gap in the theoretical analysis justifying the validity'' of the method in \citeauthor{meinshausen2006false}. We solve this by considering candidate bounds which are independent of the data, as will be explained. Moreover, we obtain a large class of methods, providing  more freedom to choose  power properties. Further, the computational complexity of the iterative method can be tuned by a user-defined parameter. For a specific choice of the parameter, the computational complexity  is linear in the number of hypotheses. In some cases, the iterative method is computationally infeasible. Hence we suggest an approximation of this procedure. The approximation method maintained the nominal error rate in all our simulation scenarios.

This paper is built up as follows. Section \ref{secsinglestep} introduces single-step procedures, including the method in  \citet{meinshausen2006false}. In Section \ref{seciter} the iterative method is presented. The various methods are compared using simulations and real data in Sections \ref{secsimsM}  and \ref{secdataM} respectively.

\section{Single-step procedures} \label{secsinglestep}

\subsection{Setting and notation} \label{setting}

Let $X$ be random  data, taking values in a sample space $\Omega$. Consider hypotheses $H_1,...,H_m$  with corresponding \emph{p}-values $P_i: \Omega\rightarrow [0,1]$, $1\leq i \leq m$.  
We will often suppress the dependence on $X$ in the notation, e.g. $P_i$  is short  for $P_i(X)$. Without loss of generality we assume that $P_1\leq ...\leq P_m$. Write $\N=\{1\leq i \leq m: \text{} H_i \text{ is true}\}$, let $n=\#\N$  (where `$\# S$' denotes the cardinality of $S$) \color{black} and  let $Q=(Q_1,...,Q_n)$  \color{black} be the sorted vector $(P_i: i\in \N)$, assuming $\N\neq\emptyset$ for convenience.

Let $\alpha\in [0,1)$ and $\co \subseteq [0,1]$ be independent of the data. The set $\co$ contains the \emph{p}-value thresholds of interest. The post-hoc chosen thresholds need to be picked from this set. Choosing $\co$ large provides much freedom in choosing the threshold post hoc,  but choosing $\co$ small generally  provides more power.    

For $t\in \co$ define
$\R= \mathcal{R}(t)=\{1\leq i \leq m:  P_i \leq t \}.$ This is the set of indices of the rejected hypotheses if each hypothesis $H_i$ is rejected when $P_i\leq t$. 
Write
$ R=\#\mathcal{R}$ and let $V=\#  (\mathcal{R} \cap \mathcal{N})$ be  the number of false positives. 
Note that $\R$ and $V$ depend on the data, but $\N$ does not.  Further, we have
$FDP(t)=V(t)/R(t),$
which is defined as $0$ when $R(t)=0$.

\subsection{Confidence envelopes} \label{secce}

The aim of this paper is to derive as small as possible simultaneous confidence bounds for the FDP.
This is equivalent to deriving as small as possible \emph{confidence envelopes}, which we define similarly to \citet{genovese2006exceedance}. In  \citet{meinshausen2005lower} these are referred to as bounding functions.

\begin{definition} \label{defenvelope}
A confidence envelope is a (possibly random) function $B:\co\rightarrow \mathbb{N}$ satisfying
$$ \mathbb{P} \Big (\bigcap_{t\in \co} \big \{V(t) \leq B(t)\big \}  \Big ) \geq 1- \alpha.$$
\end{definition}

Note that with probability at least $1-\alpha$, simultaneously for all $t\in \co$, the numbers $B(t)$ are upper bounds for the numbers of false positives $V(t)$. Note that if $B(t)\geq V(t)$ and $R(t)>0$, then ${B(t)}/{R(t)}\geq FDP(t)$. Hence, from simultaneous upper bounds for $V(t)$,  simultaneous upper bounds  for $FDP(t)$ immediately follow.


Confidence envelopes can de derived from \emph{critical vectors}.

\begin{definition}
A vector  $C=(c_1,...,c_{\#C})$, $\#C\geq n$, is  a  critical vector  if 
\begin{equation} \label{defcc}
\mathbb{P}\big ( \bigcap_{i=1}^n\big \{  Q_i\geq c_i    \big \}     \big)\geq 1-\alpha.
\end{equation}
\end{definition}

Let $[\cdot]^+$ denote the positive part function.

\begin{proposition} \label{rel}
If $C$ is a critical vector,  then the map
$B:\co \rightarrow \{1,...,m\}$ defined by 
$$B(t)=\#\big \{1\leq i\leq \#C:   c_i \leq t   \big \}$$ is a confidence envelope.
In addition , $B':\co \rightarrow \{1,...,m\}$ defined by 
\begin{equation} \label{imprenv}
B'(t)=R(t)-\max\big\{\big [R(s)-B(s)\big ]^+: s\in \co, s\leq t\big\},
\end{equation}
which satisfies $B'\leq B$, is also a confidence envelope and potentially improves $B$. 
\end{proposition}


\begin{proof}
With probability at least $1-\alpha$, $Q\geq C$, and then for each $t\in [0,1]$,
$$V(t)=\#\big \{1\leq i\leq n:   Q_i \leq t   \big  \} \leq \#\big \{1\leq i\leq \#C:   c_i \leq t   \big \}=B(t).$$
Thus $B$ is a confidence envelope. 

The number of true findings $R(t)-V(t)$ is non-decreasing in $t\in \co$.
Hence the bounds
\begin{equation} 
\max\big\{\big [R(s)-B(s)\big ]^+: s\in \co, s\leq t\big\},
\end{equation}
$t\in \co$, are simultaneous $(1-\alpha)$-lower bounds for the number of true findings $R(t)-V(t)$, $t\in \co$.
Consequently \eqref{imprenv}  is a  confidence envelope. It  improves $B$ when $(R-B)^+$ is not non-decreasing. See also Section 3.2 in \citet{meinshausen2006false}.
\end{proof}

Observe that the larger $C$ is, the smaller the confidence envelope is that is obtained with Proposition \ref{rel}. Hence it is of interest to find as large as possible $C$.
The existing literature provides various critical vectors  and we can use these to construct confidence envelopes.
An example is given in the following.

\subsection{Parametric confidence envelopes}

In many practical situations, the distribution of $Q$ is such that a well-known probability inequality by \citet{simes1986improved} holds \citep{rodland2006simes}:
\begin{equation} \label{eq:simes}
\mathbb{P}\Big (\bigcap_{i=1}^n \text{ }  \{Q_i\geq i\alpha/n  \} \Big )\geq 1-\alpha.
\end{equation}
This probability equality provides a critical vector, which can be used to obtain a confidence envelope $B:[0,1]\rightarrow \{1,...,n\}$  with Proposition \ref{rel}: 
$$B(t)=\#\{1\leq i \leq n: i \alpha /n \leq t\}.$$
However, $n$ is not known, so that this envelope is unknown in practice. One can instead note that $n \leq m$ and use the confidence envelope $B$ satisfying
\begin{equation} \label{eq:simesparenvelope}
B(t)=\#\{1\leq i \leq m: i \alpha/m\leq t\}= \#\{1\leq i \leq m: i\leq mt/\alpha\} =\lfloor mt/\alpha \rfloor \wedge m. 
\end{equation}

Simes' probability inequality is not valid for all possible dependence structures of $Q$, so that the above confidence envelope cannot always be used.
 Even if Simes' probability inequality holds, the critical vector based on it can be very conservative, because the probability at \eqref{eq:simes}
can be larger than $1-\alpha$, i.e. the nominal error rate $\alpha$ is not exhausted (even under the complete null). This happens when the $Q_i$ are positively (but not perfectly) correlated. Other parametric critical vectors are also often conservative or require much stronger assumptions \citep{cai2008modified, gou2014generalized}. 
In the following we discuss nonparametric methods, which often better exhaust $\alpha$ (in particular, they fully exhaust $\alpha$ under the complete null), leading to an increase of power.
\color{black}

\subsection{Permutation framework} \label{permfr}

All nonparametric methods in this paper are based on permutations or other transformations of the data.
Let $G$ be a finite set of  transformations $g: \Omega\rightarrow \Omega$, such that $G$ is a group (in the algebraic sense) with respect to the operation of composition of transformations. 
In practice $G$ is often a group of permutation maps. 
Sometimes other groups of transformations can be used, such as rotations \citep{langsrud2005rotation, solari2014rotation} and multiplication of part of the data by $-1$ (\cite{pesarin2010permutation}, pp.\ 54 and 168).

All permutation-based procedures in this paper rely on the following assumption.
\begin{assumption} \label{exchangeabilityas}
The joint distribution of the \emph{p}-values $P_i(g(X))$ with $i\in \N$, $g\in G$, is invariant under all transformations in $G$ of $X$.
\end{assumption}
This assumption underlies many permutation-based multiple testing methods, e.g. \citeauthor{westfall1993resampling}'s max\emph{T} method (\citeyear{westfall1993resampling}), \citet{tusher2001significance},  \citet{hemerik2018false}, \citet{meinshausen2005lower} and \citet{meinshausen2006false}.
Usually this assumption means that the joint distribution of the part of the data corresponding to $\N$ should be invariant under permutation. 

In this paper random transformations from $G$ are used, which are defined as follows.

\begin{definition} \label{defrpM}
Let $g_1:=id$ be the identity in $G $ and $g_2, ..., g_w$  random elements from $G$.  The random transformations can be drawn either with or without replacement: the statements in this paper hold for both cases. If  $g_2,..., g_w$ are drawn  without replacement, then they are taken to be uniformly distributed on $G\setminus \{id\}$, otherwise uniform on $G$.
\end{definition}

For $\I\subseteq\{1,...,m\}$  and $1\leq j \leq w$, write $R_\I^j(t):=\#\{i\in \I:P_i(g_j(X))\leq t\}$, $R^j:=R^j_{\{1,...,m\}}$ and $R_\I:= R_\I^1.$

\subsection{Nonparametric confidence envelope} \label{mein}

When Assumption \ref{exchangeabilityas} is satisfied, a confidence envelope can be constructed by using the permutation distribution of the \emph{p}-values $Q$. Since by assumption this permutation distribution retains the dependence structure of these \emph{p}-values, it can be used to construct an envelope which is adapted to this  structure. 
Until now this was only done by \citet{meinshausen2006false}. We now recall this method, before uniformly improving it in Section \ref{seciter}.

Central to the method is a family of \emph{candidate envelopes}, which we define below. In \citet{meinshausen2006false} these depend on \emph{p}-values corresponding to false null hypotheses, so that the joint distribution of $Q$ and the candidate envelope picked in \citeauthor{meinshausen2006false} is not generally permutation invariant (\citeauthor{blanchard2017post} \citeyear{blanchard2017post}, p. 19, also note this).  Hence we consider candidate envelopes that are independent of the data.  An additional difference is that we include the original observation with the random permutations \citep[see e.g.][]{hemerik2017exact}. Otherwise, the method provided here is the same as the procedure in \citet{meinshausen2006false}.

Let $\mathbb{B}$ be a set of maps $\co \rightarrow \mathbb{N}$, independent of the data. Suppose that for all $B$,  $B'\in \mathbb{B}$, either $B\geq B'$ or $B'\geq B$. $\mathbb{B}$ is the family of \emph{candidate envelopes}. 
 Examples of such $\mathbb{B}$ are in Section \ref{excand}.

Meinshausen's confidence envelope (with the above adaptations) is defined as follows.

\begin{theorem} \label{mmethod}
Let
$$B^{\mathrm{m}}=\min\Bigg\{B\in \mathbb{B}: \text{ } w^{-1}\#\Big\{1\leq j \leq w: \bigcap_{t\in \co} \big\{R^j(t) \leq  B(t) \big\}  \Big\} \geq 1-\alpha \Bigg\},$$
where we assume that $\mathbb{B}$ is such that this minimum exists. Then $B^{\mathrm{m}}$  is a confidence envelope.
\end{theorem}

\begin{proof}
Let $${B}_\N=\min\Bigg\{B\in \mathbb{B}: \text{ } w^{-1}\#\Big\{1\leq j \leq w: \bigcap_{t\in \co}\big \{ R_\N^j(t) \leq  B(t) \big \}  \Big\} \geq 1-\alpha\Bigg\}.$$
It follows from the group structure of the set of transformations $G$ \citep[][Theorem 1]{hemerik2018false} that for every $1\leq j \leq w$,
$$\mathbb{P}\Big[\bigcap_{t\in \co} \big\{  R_\N(t) \leq  {B}_\N(t)\big \}\Big]  = \mathbb{P}\Big[\bigcap_{t\in \co} \big\{  R^j_\N(t) \leq  {B}_\N(t)\big \}\Big].$$
Hence this probability equals
\begin{equation} \label{errorratem}
w^{-1}\sum_{j=1}^w  \mathbb{E} \bigg ( \mathbbm{1} \bigg [   \bigcap_{t\in \co}   \big \{    R^j_\N(t) \leq  {B}_\N(t)  \big \} \bigg ] \bigg )  \geq 1-\alpha. \end{equation}
Since $R_\N=V$, this means that ${B}_\N$ is a confidence envelope.
Hence the larger function $B^{\mathrm{m}}$ is also a confidence envelope.
\end{proof}

The choice of $ \mathbb{B}$ has a crucial influence on $B^{\text{m}}$.
It is an important assumption that for all $B$,  $B'\in  \mathbb{B}$, either $B\geq B'$ or $B'\geq B$. This guarantees that ${B}_\N(t)\leq B^{\text{m}}(t)$ for all $t\in \co$.

Under mild assumptions such as continuity, the inequality \eqref{errorratem} becomes an equality. 
If all null  hypotheses are true, then $B_\N=B^{\text{m}}$. But this means that under the complete null, the probability that the confidence envelope $B^{\text{m}}$ is invalid is exactly $\alpha$. Thus, under the complete null, the method completely exhausts the nominal error rate $\alpha$, despite the unknown dependence among the \emph{p}-values.
\color{black}


\subsection{Examples of candidate envelopes} \label{excand}

We will now give some examples of families $\mathbb{B}$.
Consider $\mathbb{B}=\{B^\lambda: \lambda \in [0,\infty) \}$, where $B^\lambda:\co\rightarrow\{1,...,m\}$ satisfies
\begin{equation} \label{eq:simesenvelopes}
B^\lambda(t)= \#\{1\leq i \leq m: i \lambda  \leq t \}.
\end{equation}
Note that by Proposition \ref{rel}, $B^\lambda$ is a confidence envelope if  the vector $(\lambda ,2\lambda ,...,m\lambda  )$ is a critical vector. This vector is simply  Simes' vector multiplied by a constant.
As another example, instead of considering the candidate envelopes \eqref{eq:simesenvelopes}, one could translate (shift) them by replacing $i\lambda$ by $i\lambda-\delta$ with $\delta>0$ a small constant, e.g. $0.001$. This makes the envelope less sensitive to the smallest \emph{p}-values.
This often results in better bounds for the larger cut-offs in $\co$, as illustrated in Fig. \ref{figenv} and Section \ref{secdataM}.

If variables $U_1,...,U_m$ are independent and uniformly distributed on $[0,1]$, and $U_{(1)}\leq ... \leq U_{(m)}$ are the sorted values of these variables, then it is well known that for every $1\leq i \leq m$, $U_{(i)}$ has a beta distribution:
$$U_{(i)} \sim \text{Beta}(i, m+1-i). $$
For each $\lambda \in [0,1]$ consider the function $B^{\lambda}:\co\rightarrow \{1,...,m\}$ given by 
$$B^{\lambda}(t) =\#\{1\leq i \leq m: q^{\lambda}_i\leq t\},$$
where $q_i^{\lambda}$ is the $\lambda$-quantile of the $\text{Beta}(i, m+1-i)$ distribution. 
In Section \ref{secdataM}  we will consider  $\{B^{\lambda}: \lambda \in (0,1)\}$ as one of the sets  of candidate envelopes. 
A heuristic reason for considering this set of candidate functions is that some of them can be  similar in shape to some of the functions $t\mapsto R^j(t)$, $2\leq j \leq w$. Consequently, the resulting confidence envelopes  tend to be relatively tight. We applied the proposed families $\mathbb{B}$ to the data of section \ref{secdataM}, see Fig. \ref{figenv}.
More examples  of candidate critical vectors (and hence candidate envelopes) are in \citet{blanchard2008two}.


\begin{figure}[ht] 
\centering
\captionsetup{width=.8\linewidth}
\begin{tikzpicture}[scale=0.9]
\begin{axis}[
anchor=origin, 
width=12.5cm, height=10cm,
  axis y line = none,
  axis x line = top,
  xmin=0.001, xmax=0.02, xlabel={Number of rejections $R(t)$}, ymin=0, ymax=250, xtick={0.001,0.005,0.01,0.015,0.02}, xticklabels={449,775,957,1080,1168}
]
\end{axis}
\begin{axis}[
    anchor=origin,
                legend style={at={(12.5cm,.7)}, anchor=west, legend columns=1, draw=none},
    width=12.5cm, height=10cm,
axis x line=bottom, xlabel={Cut-off $t$}, xmin=0.001, xmax=0.02,  xtick={0.001,0.005,0.01,0.015,0.02}, xticklabels={0.001,0.005,0.01,0.015,0.02}, 
       axis y line=left, ylabel={Bounds $B(t)$}, ymin=0, ymax=250]

\addplot[dash pattern=on 9pt off 4pt , line width = 1, very thick] table[x="co",y="beta", col sep=comma,smooth]{envelopes.csv};

\addplot[smooth, dashed, line width = 1, very thick] table[x="co",y="simes", col sep=comma]{envelopes.csv};

\addplot[ line width = 1, very thick] table[x="co",y="shift", col sep=comma]{envelopes.csv};

\addplot[color=gray,  dotted, line width = 1, very thick] table[x="co",y="perm1", col sep=comma]{envelopes.csv};
\addplot[color=gray,  dotted, line width = 1, very thick] table[x="co",y="perm2", col sep=comma]{envelopes.csv};
\addplot[color=gray, dotted, line width = 1, very thick] table[x="co",y="perm3", col sep=comma]{envelopes.csv};
\addplot[color=gray,  dotted, line width = 1, very thick] table[x="co",y="perm4", col sep=comma]{envelopes.csv};
\addplot[color=gray,dotted, line width = 1, very thick] table[x="co",y="perm5", col sep=comma]{envelopes.csv};
\addplot[color=gray, dotted, line width = 1, very thick] table[x="co",y="perm6", col sep=comma]{envelopes.csv};
\addplot[color=gray, dotted, line width = 1, very thick] table[x="co",y="perm7", col sep=comma]{envelopes.csv};
\addplot[color=gray,  dotted, line width = 1, very thick] table[x="co",y="perm8", col sep=comma]{envelopes.csv};
\addplot[color=gray,  dotted, line width = 1, very thick] table[x="co",y="perm9", col sep=comma]{envelopes.csv};
\addplot[color=gray, dotted, line width = 1, very thick] table[x="co",y="perm10", col sep=comma]{envelopes.csv};

\addplot[color=gray,  dotted, line width = 1, very thick] table[x="co",y="perm11", col sep=comma]{envelopes.csv};
\addplot[color=gray,  dotted, line width = 1, very thick] table[x="co",y="perm12", col sep=comma]{envelopes.csv};
\addplot[color=gray, dotted, line width = 1, very thick] table[x="co",y="perm13", col sep=comma]{envelopes.csv};
\addplot[color=gray,  dotted, line width = 1, very thick] table[x="co",y="perm14", col sep=comma]{envelopes.csv};
\addplot[color=gray,dotted, line width = 1, very thick] table[x="co",y="perm15", col sep=comma]{envelopes.csv};
\addplot[color=gray, dotted, line width = 1, very thick] table[x="co",y="perm16", col sep=comma]{envelopes.csv};
\addplot[color=gray, dotted, line width = 1, very thick] table[x="co",y="perm17", col sep=comma]{envelopes.csv};
\addplot[color=gray,  dotted, line width = 1, very thick] table[x="co",y="perm18", col sep=comma]{envelopes.csv};
\addplot[color=gray,  dotted, line width = 1, very thick] table[x="co",y="perm19", col sep=comma]{envelopes.csv};
\addplot[color=gray, dotted, line width = 1, very thick] table[x="co",y="perm20", col sep=comma]{envelopes.csv};

\addplot[color=gray,  dotted, line width = 1, very thick] table[x="co",y="perm21", col sep=comma]{envelopes.csv};
\addplot[color=gray,  dotted, line width = 1, very thick] table[x="co",y="perm22", col sep=comma]{envelopes.csv};
\addplot[color=gray, dotted, line width = 1, very thick] table[x="co",y="perm23", col sep=comma]{envelopes.csv};
\addplot[color=gray,  dotted, line width = 1, very thick] table[x="co",y="perm24", col sep=comma]{envelopes.csv};
\addplot[color=gray,dotted, line width = 1, very thick] table[x="co",y="perm25", col sep=comma]{envelopes.csv};
\addplot[color=gray, dotted, line width = 1, very thick] table[x="co",y="perm26", col sep=comma]{envelopes.csv};
\addplot[color=gray, dotted, line width = 1, very thick] table[x="co",y="perm27", col sep=comma]{envelopes.csv};
\addplot[color=gray,  dotted, line width = 1, very thick] table[x="co",y="perm28", col sep=comma]{envelopes.csv};
\addplot[color=gray,  dotted, line width = 1, very thick] table[x="co",y="perm29", col sep=comma]{envelopes.csv};
\addplot[color=gray, dotted, line width = 1, very thick] table[x="co",y="perm30", col sep=comma]{envelopes.csv};

\addplot[color=gray,  dotted, line width = 1, very thick] table[x="co",y="perm31", col sep=comma]{envelopes.csv};
\addplot[color=gray,  dotted, line width = 1, very thick] table[x="co",y="perm32", col sep=comma]{envelopes.csv};
\addplot[color=gray, dotted, line width = 1, very thick] table[x="co",y="perm33", col sep=comma]{envelopes.csv};
\addplot[color=gray,  dotted, line width = 1, very thick] table[x="co",y="perm34", col sep=comma]{envelopes.csv};
\addplot[color=gray,dotted, line width = 1, very thick] table[x="co",y="perm35", col sep=comma]{envelopes.csv};
\addplot[color=gray, dotted, line width = 1, very thick] table[x="co",y="perm36", col sep=comma]{envelopes.csv};
\addplot[color=gray, dotted, line width = 1, very thick] table[x="co",y="perm37", col sep=comma]{envelopes.csv};
\addplot[color=gray,  dotted, line width = 1, very thick] table[x="co",y="perm38", col sep=comma]{envelopes.csv};
\addplot[color=gray,  dotted, line width = 1, very thick] table[x="co",y="perm39", col sep=comma]{envelopes.csv};
\addplot[color=gray, dotted, line width = 1, very thick] table[x="co",y="perm40", col sep=comma]{envelopes.csv};

\addplot[color=gray,  dotted, line width = 1, very thick] table[x="co",y="perm51", col sep=comma]{envelopes.csv};
\addplot[color=gray,  dotted, line width = 1, very thick] table[x="co",y="perm52", col sep=comma]{envelopes.csv};
\addplot[color=gray, dotted, line width = 1, very thick] table[x="co",y="perm53", col sep=comma]{envelopes.csv};
\addplot[color=gray,  dotted, line width = 1, very thick] table[x="co",y="perm54", col sep=comma]{envelopes.csv};
\addplot[color=gray,dotted, line width = 1, very thick] table[x="co",y="perm55", col sep=comma]{envelopes.csv};
\addplot[color=gray, dotted, line width = 1, very thick] table[x="co",y="perm56", col sep=comma]{envelopes.csv};
\addplot[color=gray, dotted, line width = 1, very thick] table[x="co",y="perm57", col sep=comma]{envelopes.csv};
\addplot[color=gray,  dotted, line width = 1, very thick] table[x="co",y="perm58", col sep=comma]{envelopes.csv};
\addplot[color=gray,  dotted, line width = 1, very thick] table[x="co",y="perm59", col sep=comma]{envelopes.csv};
\addplot[color=gray, dotted, line width = 1, very thick] table[x="co",y="perm60", col sep=comma]{envelopes.csv};


\end{axis}
\fill[white] (10,-1.1) rectangle (11,-0.5);
\end{tikzpicture}
\caption{For three different families $\mathbb{B}$, resulting $90\%$-confidence envelopes are shown for cut-offs in  $\mathbb{T}=[0.001,0.02]$ (van de Vijver data, see Section \ref{secdataM}). Moreover, for some of the permuted versions of the data, the corresponding numbers of rejections $R^j(t)$ are shown (dotted). Each confidence envelope lies above $90\%$ of these curves. The envelopes are based on the following families $\mathbb{B}$:  Simes-type (small dashes), shifted Simes-type (solid) and beta distribution-based (large dashes).}
\label{figenv}
\end{figure}
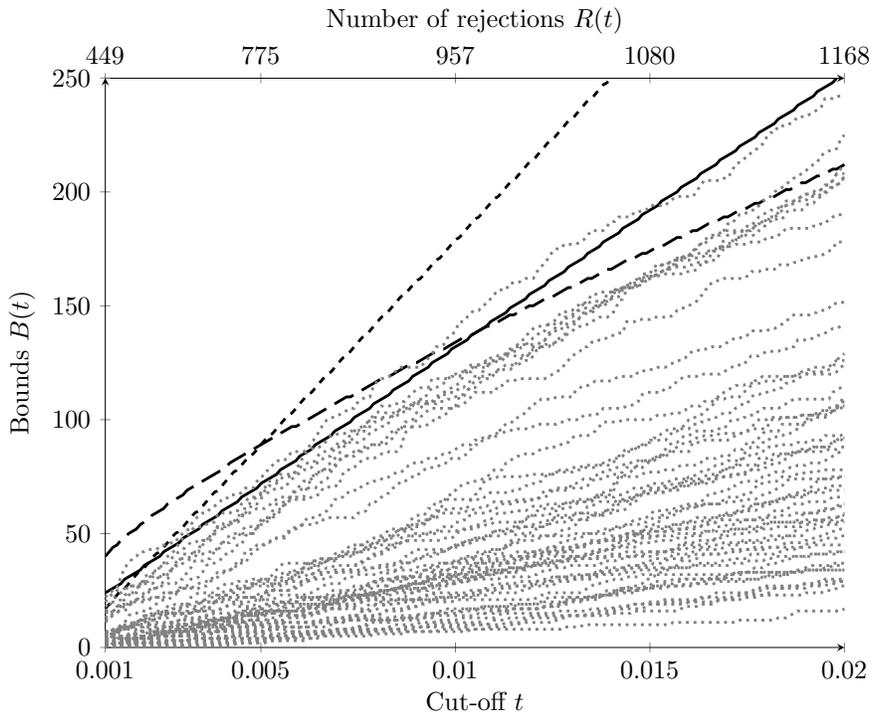


We now show that two existing multiple testing methods, Significance Analysis of Microarrays (SAM) \citep{tusher2001significance, hemerik2018false} and the single-step max\emph{T} method by \citet{westfall1993resampling}, are special cases of the general method at Theorem \ref{mmethod}.
These methods essentially only differ with respect to the family $\mathbb{B}$ of candidate envelopes on which they are based.  

Let $c\in \co$ be independent of the data. Consider the family of candidate  envelopes $\mathbb{B}=\{B^0,B^1,...,B^m\}$, where for $\lambda\in \{0,...,m\}$, $B^{\lambda}:\co\rightarrow \{0,...,m\}$ satisfies
\[
    B^{\lambda}(t)= 
\begin{cases}
    \lambda,& \text{if } t \leq c\\
    m,              & \text{otherwise.}
\end{cases}
\] 
Note that if Theorem \ref{mmethod} is applied based on these candidate functions, then the resulting upper bound $B^\text{m}(c)$ for $V(c)$ is simply the $\lceil(1-\alpha)w\rceil w^{-1}$-quantile of the values $R^j(c)$, $1\leq j \leq w$. This is precisely the (most basic) confidence bound in \citet{hemerik2018false}. That paper extends the Significance Analysis of Microarrays methodology by \citet{tusher2001significance}, who estimate the FDP using permutations, but do not provide a confidence bound.

Consider the family $\mathbb{B}=\{B^{\lambda}: \lambda\in [0,1]\}$, where $B^{\lambda}:[0,1]\rightarrow \{0,...,m\}$ is defined by
\[
    B^{\lambda}(t)= 
\begin{cases}
    0,& \text{if } t < \lambda \\
    m,              & \text{otherwise.}
\end{cases}
\] 
Applying Theorem \ref{mmethod}  to these candidate envelopes results in the upper bound $B^\text{m}=B^{\lambda'}$, where $\lambda'$ is the $\alpha$-quantile  of the values $\min_{1\leq i \leq m}P_i(g_jX)$, $1\leq j \leq w.$ The bound $B^{\lambda'}(t)$ equals zero for $t < \lambda' $, which means that the family-wise error rate is strongly controlled if the hypotheses $\{1\leq i \leq m : P_i < \lambda'\}$ are rejected. This is exactly the set of hypotheses that the single-step max\emph{T} method rejects \citep{westfall1993resampling}.
Moreover, using the iterative method in Section \ref{seciter}, the step-down max\emph{T} method  can be obtained.


\section{Iterative method} \label{seciter}

\subsection{Introduction}

The method of Theorem \ref{mmethod} can be uniformly  improved with a method by \citet{goeman2011multiple}, which is related to closed testing \citep{marcus1976closed}. 
Unless the number of hypotheses is very small  (less than 20), \color{black} this method is usually computationally infeasible in the context of this paper. Hence we  discuss this method in the Supplementary Material. There we also prove that the method of \citet{goeman2011multiple} is equivalent to that in \citet{genovese2004stochastic, genovese2006exceedance}.

Below we will derive  a general, iterative method for improvement of the basic confidence envelope $B^{m}$.
In each iteration step, the method uses an FDP upper bound obtained in the previous step. Some existing family-wise error rate controlling methods, where in each step the rejections from the previous steps are used, are special cases of this method \citep[e.g.][]{holm1979simple,  westfall1993resampling}.

For each nonempty $\I\subseteq \{1,...,m\}$, consider a function $B_\I:\co\rightarrow \mathbb{N}$, such that $B_{\I}\leq B_{\J}$ whenever $\I \subseteq \J \subseteq \{1,...,m\}$ and such that
$B_{\N}$ is a confidence envelope.

In particular, we can consider
\begin{equation} \label{eq:BInonp}
B_\I=\min\Bigg\{B\in  \mathbb{B}: \text{ } w^{-1}\#\Big\{1\leq j \leq w: \bigcap_{t\in \co}\big \{ R_\I^j(t) \leq  B(t) \big \} \Big\} \geq 1-\alpha \Bigg\}.
\end{equation}
For this definition of $B_\I$, $B_{\{1,...,m\}}$ coincides with $B^\text{m}$. Thus, intuitively, $B^\text{m}$ is an envelope which takes into account the worst-case scenario that $\N=\{1,...,m\}$. If instead it were known (hypothetically) that $H_1$ were false, for example, then $B_{\{2,...,m\}}$ could be used as a confidence envelope.
The iterative method below uniformly improves $B_{\{1,...,m\}}$.

\subsection{Exact method} \label{im} \label{nip}

We now define the \emph{iterative method}.
\begin{theorem} \label{thmit}
Fix some $s\in \co$. Let $B^0=B_{\{1,...,m\}}$ and for $i\in \mathbb{N}$ iteratively define
$$B^{i+1}(t) = \max\{  B_{\K^c}(t): \K\in \R(s), \#\K=  R(s)- B^i(s)    \}.$$
It holds that
$B^0 \geq B^1\geq...$ and from a certain $i\in \mathbb{N}$, $B^i = B^{i+1}=...$.
The function $B^{\text{it}}=\min_{i\in \mathbb{N}} B^i$ is a confidence envelope.
\end{theorem}
\begin{proof}
Define the event
$$ E :=\bigcap_{t\in \co}  \Big \{ V(t) \leq B_\N(t)     \Big \}.$$
Assume $E$ holds.
For $i=0$ we have $V(t) \leq B^i(t)$  for all $t\in \co$. Let $i\in \mathbb{N}$ and suppose that the same holds for this $i$.
Then there is a set $\K\subseteq \R(s)$ with $\#\K=  R(s)- B^i(s)$ such that $\N\subseteq \K^c=\{1,...,m\}\setminus \K$. 
Thus $B_{\K^c}\geq B_{\N}$. In practice it is not known for which set $\K$ this holds, but we know that $V(t)\leq B_\N(t) \leq  B^{i+1}(t)$ for all $t\in \co$.

Thus, by induction, under $E$, for all $i\in \mathbb{N}$, $V(t) \leq  B^{i+1}(t)$ for all $t\in \co$.
Since $\mathbb{P}(E)\geq 1-\alpha$, it follows that $B^{\text{it}}$ is a confidence envelope.
\end{proof}

In many practical situations convergence of the  decreasing sequence of integers  $B^0,B^1,...$ (which happens as soon as two consecutive values are equal) \color{black} is reached after only a few steps.

This iterative procedure can be modified in several ways. 
Above, in the $i$-th step $B^{i}$ is computed using one cut-off $s$.
A better bound could be obtained by doing this for many $s\in \co$ and letting $B^i$ be the pointwise minimum of all the improved bounds obtained. The resulting bound $B^i$ is still  valid under $E$.
Including such steps however increases the computational burden, so it may be better to use the  method based on one cut-off $s$ as described above.


When  $B_{\I}$ is  defined as \eqref{eq:BInonp}, we will refer to the iterative method as the \emph{nonparametric iterative method}. This method is a uniform improvement of Meinshausen's envelope $B^{\text{m}}$  in Section \ref{mein}, if  the same family $\mathbb{B}$ is used.

The  nonparametric iterative procedure is much faster than the corresponding procedure  based on closed testing \citep{goeman2011multiple}. However,  it can still be computationally infeasible, since performing one step of this procedure requires calculating a maximum of a set of size
$$\binom{R(s)}{B^i(s)}.$$
This consideration may be used to guide the choice of $s$. In particular, $s$ can be taken such that $B^i(s)$ is small.
Note that $s$ can even  be taken such that $B^i(s)=0$, leading to a very fast method. 
In that case, $B^1(t)=B_{R^c(s)}(t)$, $t\in \co$, which can considerably improve the single-step bound $B^0(t)$ if $R(s)$ is large.

\subsection{Approximation method} \label{secap}

We suggest a method for approximating the confidence envelope $B^{\text{it}}$,  for cases where the iterative method computationally infeasible. The approximation method is feasible when there are many thousands of hypotheses. \color{black}
In the iterative method, computing any $B^i(s) $ requires determining a maximum of a potentially very large set. The approximation method  computes the maximum over a smaller, random subset, to limit the computation time.

Write $\widehat{B}^0:=B_{\{1,...,m\}}$ and for $i=1,2,...$ iteratively compute
$\widehat{B}^{i}(s):= \max\{  B_{\K^c}(s): \K\in \rset^{i}   \},  $
where $\rset^i$ is some large  random subcollection of 
$\{\K\in \R(s):\text{ } \#\K=  R(s)-\widehat{B}^{i-1}(s)\}$.
Recall that if $B^i(s)=B^{i+1}(s)$, then $B^{i+1}=B^{\text{it}}$. Hence if $\widehat{B}^{i}(s)=\widehat{B}^{i+1}(s)$, then $\widehat{B}^{i+1}(t)=\max\{  B_{\K^c}(t): \K\in \rset^{i+1}    \} $ can be seen as an estimate of $B^{\text{it}}(t)$, $t\in \co$.


Observe that for $\#\rset^1 \rightarrow \infty$, almost surely $\widehat{B}^{1}(s)\rightarrow B^{1}(s)$ (assuming $\rset^1$  is uniformly sampled). Similarly, if $\#\rset^1,..., \#\rset^{i+1} \rightarrow \infty$, then $\widehat{B}^{i}(s)\rightarrow B^i(s)$ and hence $\widehat{B}^{i+1} \rightarrow {B}^{i+1}$ uniformly.
Thus, the approximation method becomes exact as the number of combinations that  it checks increases to infinity.
For finite $\#\rset^i$, the approximation method may potentially be anti-conservative, but this was not the case in our simulation settings.

For $\K\in \rset^{i}$, the time needed to compute $\widehat{B}^{i}(s)$ is linear in $m$, so that the computation time for the approximation method is also linear in $m$.
\color{black}

\section{Simulations} \label{secsimsM}

\subsection{Simulation setting}

To compare the methods of this paper, we applied them to simple simulated data. In Section \ref{simit} the performance of the iterative method as compared to the single-step method is investigated. In Section \ref{simap} the validity of the approximation method is discussed. 
See the data analysis in Section \ref{secdataM} for a comparison of our nonparametric methods with the parametric variants.  

The simulated data matrix was the $20\times m$-matrix ${X}={X}'+{Z}.$
It can be seen as representing $m$ measurements for $20$ persons. 
Here ${X}'$ is a $20\times m$-matrix of independent normally distributed variables with variance 1. 
For some $0\leq F \leq m$, in the first $F$ columns of ${X}$ the first $10$ entries had mean $1.5$ and all other entries  had mean $0$.
The matrix ${Z}$, which determined the correlation structure of ${X}$, is defined by ${Z}_{ji}:=s_i Z_j$, where
$s_i=1$ for $i$ odd and  $s_i=-1$ for $i$ even. Here each $Z_j$ is independent and normally distributed with mean $0$ and  standard deviation $ {\sigma_Z}$.
For $1\leq j \leq 20$ and $1\leq i< i' \leq m$ note that the correlation is $\rho({X}_{ji},{X}_{ji'})=\pm  (\sigma_Z^2)/(1+ {\sigma_Z}^2).$

For each $1\leq i \leq m$, let $H_i$ be the null hypothesis that ${X}_{1,i}...,{X}_{20,i}$ are independent and standard normally distributed. 
Thus the  fraction of true null hypotheses was $\pi_0:=(m-F)/m$.
For each $H_i$,  $P_i$ was defined as the \emph{p}-value from a two sided  t-test comparing the first $10$ individuals with the last $10$.

As  $G$ we took all $20!$ permutations of cases and controls. In all the simulations  we used $w=100$, i.e. each time we drew $99$ random permutations (with replacement) and added the identity.  For larger $w$ similar results are obtained  \citep[see also][]{marriott1979barnard}.
We took $\alpha=0.1$. The values of $m$, $\pi_0$ and  $|\rho|$ are specified per case below.

\subsection{Performance of the iterative method} \label{simit}

We now illustrate that the nonparametric single-step method of Section \ref{mein}  (Theorem \ref{mmethod}) \color{black} is improved by the corresponding iterative procedure (Section \ref{nip}). We took $m=50$ since the iterative method is not always feasible for large numbers of hypotheses.  When the number of hypotheses exceeds a few hundred, the user will usually need to use the approximation method (Section \ref{secap}). 

We will see that the improvements with the iterative method are limited, which is due to the small $m$.  For larger $m$, larger improvements are obtained, see Section \ref{simap}. \color{black} We took $\co=[0.001,0.01]$. 
As candidate envelopes we took
 $B^{\lambda}(t)=\#\{ 1\leq i \leq m: i\lambda -0.001 \leq t   \}$, $\lambda\in [0,\infty)$. In the iterative method $s$ was taken to be $0.005$. The iterative method was always terminated after three steps, when it had usually converged.

We estimated  the expected values of the FDP bounds  (which are of the form $B(t)/R(t)$) \color{black} for different values of $\pi_0$ and $|\rho|$  (where $|\rho|$ depends on $\sigma_Z^2$). \color{black} Above the columns the cut-offs that were used, are shown. For example, a cut-off of $0.01$ means that all hypotheses with \emph{p}-values smaller than $0.01$ were rejected.

The results are shown in Table \ref{table:it}. 
The simulations  in the setting $\pi_0=0.4, |\rho|=0.5$ took the longest, with a few seconds per analysis on average on a standard PC, i.e.  about half an hour for 1000 simulations.
 Each estimate is based on 1000 simulations,  so that for each setting the standard error of the mean difference between the two bounds is smaller than $9\cdot 10^{-4}$. Note that regardless of the standard error, the difference in performance is significant, since by construction the iterative method provides a bound at least as small as the bound from the single-step method. \color{black}
 
For the cut-off $0.001$, the upper bounds were usually zero. This is not surprising: for such a small cut-off, it is indeed very likely that there are no false positives (given the limited number of hypotheses, $m=50$). \color{black} The improvement with the iterative method was largest when $\pi_0$ was small, i.e. when there were many false  null \color{black} hypotheses. When $m$ was larger, bigger improvements were obtained, see Sections \ref{simap} and \ref{secdataM}. \color{black}

\begin{table}[!ht] \normalsize  
\caption{ Comparison of the single-step method with the \emph{iterative method} (italic). The values shown are the estimated expected values of the bounds. The values above the columns indicate the cut-offs.   } 
\begin{center}
    \begin{tabular}{ l  l l  l  l l l  l }    
\hline \\[-0.4cm]
& & \multicolumn{6}{l}{\qquad \qquad \qquad \qquad    Cut-off} \\ \cline{3-8} 
$\pi_0$ &  $|\rho|$ &  \multicolumn{2}{l}{\qquad  0.001}  & \multicolumn{2}{l}{\qquad  0.005}  & \multicolumn{2}{l}{\qquad  0.01}  \\ \hline 
$0.8$   &  0   \qquad \quad &  0.000 & $\emph{0.000}$ &  0.172 & $\emph{0.170}$ & 0.306 & $\emph{0.302}$  \\ 
$0.8$   &  0.5  \qquad \quad & 0.000 & $\emph{0.000}$ & 0.216 & $\emph{0.216}$ & 0.438 &  $\emph{0.435}$  \\ 
$0.6$   &  0    \qquad \quad & 0.000 & $\emph{0.000}$ & 0.104 & $\emph{0.101}$ & $0.187$ & $\emph{0.179}$   \\ 
$0.6$   &   0.5   \qquad \quad & 0.002 & $\emph{0.002}$ & 0.201  & $\emph{0.198}$ & 0.323 &  $\emph{0.318}$  \\ 
$0.4$   & 0    \qquad \quad &   0.000 & $\emph{0.000}$ & 0.073  & $\emph{0.067}$ & 0.131 & $\emph{0.117}$   \\
$0.4$   & 0.5    \qquad \quad &  0.001 & $\emph{0.001}$  & 0.148 & $\emph{0.144}$ & 0.233 & $\emph{0.228}$ \\ \hline    
    \end{tabular}
\label{table:it}
\end{center}
\end{table}

\subsection{Performance of the approximation method} \label{simap}

The approximation method is much faster than the iterative method and can be used when there are many thousands of hypotheses. \color{black}
We first compare the approximation method (Section \ref{secap}) with the iterative method. This is done in the settings of Section \ref{simit} with $m=50$.
Write $\overline{FDP}_{\text{it}}(t)=B^{\text{it}}(t)/R(t)$ and let $\overline{FDP}_{\text{ap}}$ be the estimate of $\overline{FDP}_{\text{it}}$ obtained with the approximation method. Again three iteration steps were used.

We recorded the average difference between the iterative and approximate bound, $|\overline{FDP}_{\text{it}}- \overline{FDP}_{\text{ap}} |$. In each step of the approximation method  100 random  combinations were used (uniformly drawn with replacement), i.e. $\#\rset^1=\#\rset^2=100$. Despite this limited number of random combinations, the approximations were already rather good: in all settings the mean value of  $|\overline{FDP}_{\text{it}}- \overline{FDP}_{\text{ap}} |$ was at most 0.0008 (results not shown).
This means that  the difference $\overline{FDP}_{\text{it}}- \overline{FDP}_{\text{ap}} $ was usually $0$ and sometimes slightly larger.
Naturally, when $\#\rset^1$ and $\#\rset^2$ were taken larger, the approximations were even better.

Note that whether $\overline{FDP}_{\text{ap}}$ closely approximates  $\overline{FDP}_{\text{it}}$ is irrelevant for our purposes, as long as 
$$\mathbb{P} \Big (\bigcap_{t\in \co} \{FDP(t) \leq \overline{FDP}_{\text{ap}}(t)\} \Big ) \geq 1-\alpha.$$
This was always the case in the settings of sections \ref{simit} and in the analogous setting with $m=1000$ (results not shown).

Table \ref{table:basicvsap} shows the improvement with the approximation method relative to the single-step method in the settings with $m=1000$. The improvement is largest for small $\pi_0$ and $|\rho|$. It can be seen that the bounds do not always increase with the cut-off, which is due to the choice of $\mathbb{B}$ and the fact that $R(t)$ increases with $t\in \co$. The computation time was about 15 seconds per analysis on average, i.e. a few hours per setting for 1000 simulations.\color{black}


\begin{table}[!ht] \normalsize  
\caption{ Comparison of the single-step method with the \emph{approximation method} (italic). The values shown are the estimated expected values of the bounds. Each estimate is based on 1000 simulations,  so that for each setting and cut-off the standard error of the mean difference between the two bounds is smaller than $5\cdot 10^{-4}$.} 
\begin{center}
    \begin{tabular}{ l  l l  l  l l l  l }    
\hline \\[-0.4cm]
& & \multicolumn{6}{l}{\qquad \qquad \qquad \qquad    Cut-off}\\ \cline{3-8} 
$\pi_0$ &  $|\rho|$ &  \multicolumn{2}{l}{\qquad 0.001}  & \multicolumn{2}{l}{\qquad 0.005}  & \multicolumn{2}{l}{  \qquad 0.01}  \\ \hline 
$0.8$   &  0   \qquad \quad &  0.045 & $\emph{0.045} $&  0.086 & $ \emph{0.082} $& 0.132 &  $\emph{0.127} $ \\ 
$0.8$   &  0.5   \qquad \quad &  0.346 &  $\emph{0.344}$ &  0.346 &  $\emph{0.343}$ & 0.418 &  $\emph{0.414} $\\ 
$0.6$   &  0    \qquad \quad &  0.025 & $ \emph{0.022} $&  0.048 & $ \emph{0.041}$ & 0.075 & $ \emph{0.064} $ \\ 
$0.6$   &   0.5   \qquad \quad  &  0.194 & $ \emph{0.189}$ &  0.188 & $ \emph{0.182}$ & 0.227 & $ \emph{0.219} $ \\ 
$0.4$   & 0    \qquad \quad &  0.020 &  $\emph{0.014} $&  0.037 & $ \emph{0.026}$ & 0.058 & $ \emph{0.041}$ \\ 
$0.4$   & 0.5    \qquad \quad &   0.144 &  $\emph{0.137}$ &  0.132 & $ \emph{0.124} $& 0.160 &$  \emph{0.150}$  \\  \hline    
    \end{tabular}
\label{table:basicvsap}
\end{center}
\end{table}


\section{Data analysis} \label{secdataM}

To illustrate and compare the methods in this paper, we apply them to a dataset by van de Vijver, available in the R package \emph{cancerdata}.
The dataset contains survival data on 295 cancer patients. For each individual, time to metastasis (if any), survival and the follow-up time are known. Moreover, for each individual the expression rates of 4928 genes  are known (we excluded 20 genes with missing values).

We consider hypotheses $H_i$, $1\leq i \leq 4928$, where $H_i$ is the hypothesis   that metastasis-free survival is not associated with the expression rate of gene $i$.   
The set $G$ of transformations used was the collection of all $295!$ maps that  permute (as pairs) the follow-up times and metastasis-free survival indicators of the $295$ individuals. Here we took $w=100$, i.e. we used $99$ random permutations and included the original data. A good feature of our methods is that they have proven validity  if a finite number of random permutations are used. \color{black}
Taking $w$ larger leads to similar results \citep[see also][]{marriott1979barnard}.

For each gene separately, we fitted a Cox proportional hazards model with this gene as the only covariate. We then computed a score test \emph{p}-value for association with metastasis-free survival. The validity of the following nonparametric methods does not rely on the validity of the assumptions of the Cox model. Indeed, the \emph{p}-values need not be exact as long as for each permutation they are defined in the same way. (Note that in the proofs, we do not require the null \emph{p}-values to be exactly  standard uniform.)

Note that we require Assumption \ref{exchangeabilityas} to hold, which says that the joint distribution of  the gene expression rates corresponding to $\N$ (rather than just the marginals) should be independent of metastasis-free survival. This property is implied if we assume the validity of the following directed acyclic graph:
$$Y \leftarrow  E \rightarrow F \leftarrow N \rightarrow T,$$ where
$Y$ is the survival outcome;
$E$ is all survival-relevant (latent) biology;
$F$ are the variables (genes) for which the null is false;
$N$ is all survival-irrelevant (latent) biology and
$T$ are the variables (genes) for which the null is true.
Here arrows indicate conditional dependencies.  The main assumption that this model makes, is independence of the joint distributions of the survival-related biology $E$ and the null variables $T$.
This assumption implies the validity of  Assumption \ref{exchangeabilityas}. \color{black}

We applied eight different methods to the data. 
With each method we obtained simultaneous FDP bounds. The set $\co$ of cut-offs is specified per case. We took $\alpha=0.1$, so that the  simultaneous bounds are valid with probability at least $90\%$.   For three cut-offs, the bounds are shown in Table \ref{table:families}. 
Here the rows correspond to the methods. The first two methods are parametric and the other methods  are based on permutations.
We will now discuss  the methods in the order of the rows of Table \ref{table:families} and compare the results.

\begin{table}[!ht] \normalsize  
\caption{Comparison of eight methods. For three cut-off values, simultaneous $90\%$-confidence  upper bounds for the FDP are shown.} 
\begin{center}
    \begin{tabular}{ l l l l l }    
\hline \\[-0.4cm]

  & & \multicolumn{3}{l}{\qquad \quad Cut-off}\\ \cline{3-5} 
 Method & $\co$  &    0.001 & 0.005 & 0.01  \\ \hline 
1: Parametric (Simes)       \qquad &$[0,1]$  \quad & 0.096 & 0.280 & 0.409 \\ 
2: Parametric (no Simes)       \qquad & $[0,1]$ \quad & 0.552 & 0.741 & 0.790 \\ 
3: Beta      \qquad &  $[0.001,0.01]$    \quad & 0.076 & 0.101 & 0.125 \\
4: Simes-type      \qquad & $[0.001,0.01]$   \quad  &   0.038  & 0.115 & 0.186 \\ 
5: Simes-type        \qquad & $[0,1]$  \quad  &  0.143  & 0.397 & 0.512 \\ 
6: Simes-type (shift)        \qquad &  $[0,1]$  \quad  &   0.053   & 0.093 & 0.137 \\ 
7: Iterative         \qquad & $[0.001,0.01]$  \quad   &  0.033   & 0.098 &  0.158  \\ 
8: Iterative (shift)     \qquad & $[0.001,0.01]$   \quad  &  0.047 &  0.085 & 0.125 \\  
Number of rejections      \qquad &    \quad  & 449 &  775  & 957 \\ \hline  
    \end{tabular}
\label{table:families}
\end{center}
\end{table}

\begin{enumerate}
\item The first  method used (see the first row of Table \ref{table:families}) is the parametric closed testing-based method with local tests based on Simes' probability inequality \citep[see][or the Supplementary Information]{goeman2011multiple, meijer2017simultaneous}
The bounds were obtained using the \emph{pickSimes} function in the R package \emph{cherry}. Note that Simes' probability inequality is an assumption, which cannot be guaranteed to hold.

\item The second method  is the same as the first, except that the local tests are not based on Simes' probability inequality, but on a different probability inequality  \citep[by][]{hommel1983tests} that always holds. Since this method uses no assumption on the dependence structure of the \emph{p}-values, the bounds obtained are much larger than those from the first method. 

\item Thirdly, we applied the nonparametric single-step method (Section \ref{mein}), where the family $\mathbb{B}$ of candidate envelopes was based on the beta distribution as explained in Section \ref{mein}. We took $\co=[0.001,0.01]$. This is arbitrary, but represents a reasonable range of thresholds of interest. Note that the obtained bounds are better than those derived with the two parametric methods. The reason for this is twofold. First, permutations were used such that the method took into account the dependence structure of the data. Second, bounds were not computed for all possible sets of hypotheses, but only for cut-offs in $\co$. The nonparametric method effortlessly adapts to  $\co$, while  there is no known parametric method that does this.

\item  Methods 3 and 4 are the same, except that in method 4  $\mathbb{B}$ was taken to be the family of Simes-type candidate envelopes given at \eqref{eq:simesenvelopes}. These candidate envelopes $B^{\lambda}(t)$ are relatively small for small cut-offs $t$, compared to the family based on the beta distribution. Consequently it is seen in the table that the bound for method 4 is better than that for method 3 when the cut-off is small (0.001). When the cut-off is larger (0.01) it is the other way around. 

\item Methods 4 and 5 are the same, except that in method 5 $\co=[0,1]$ was taken. Since the bounds are now uniform over a larger set, they are larger than those obtained with method 4 for all cut-offs in $[0.001,0.01]$. 

\item Method 6 is the same as method 5, except that  in the definition of the candidate envelopes $B^{\lambda}(t)$ at \eqref{eq:simesenvelopes}, $\lambda i$ is replaced by $\lambda i- 0.001$.  By comparing rows 5 and 6 in the table, it can be seen that this leads to much better (i.e. smaller) upper bounds for many cut-offs (but not for  cut-offs very close to zero, which are now shown in the table). The reason is that method 5 is too sensitive to the smallest \emph{p}-values, whose $0.1$-quantile is quite small relative to their mean \citep[see also][Section 4.3]{blanchard2017post}.
(The shift of $-0.001$ is somewhat arbitrary, but compared to other shifts it provided a good trade-off between obtaining good bounds for the small and the large cut-offs.)

\item Methods 7 and 8 are variants of the approximation of the iterative method as defined in Section \ref{secap}. The first step of method 7 coincides with method 4, and then additional iterative steps were performed as in Section \ref{secap} (with $s=0.005$ and $\#\rset^i=1000$). Note the uniform  improvement in comparison to method 4. The computation time was about $40$ minutes on a standard PC. Note however that, as stated in Section \ref{secap}, the computation time is only linear in the number of hypotheses. \color{black} 

\color{black}

\color{black}

\item Method 8 coincides with method 7, except that the family $\mathbb{B}$ was shifted as in method 6. Compared to method 7, this  improves the upper bounds for the larger cut-offs, as before.

\end{enumerate}
The first conclusion to be drawn from these results, is that the a priori chosen family $\mathbb{B}$ of candidate  envelopes has a large impact on the resulting confidence envelope.  The second conclusion is that
when $\co$ becomes smaller than $[0,1]$, the bounds from the nonparametric method can improve substantially, while there is no known parametric method that adapts to $\co$.

Although the performance of the methods strongly relies on the family $\mathbb{B}$, it should be noted that one family of candidate envelopes cannot be uniformly better than any other. For example, for very small cut-offs (not shown) method 6 was outperformed by method 5.

Precisely because the family $\mathbb{B}$ has a large impact on the results, it should be emphasized that this set must be chosen before looking at the data. In the opposite case, the family $\mathbb{B}$ would be selected based on the data in such a way that the results are as attractive as possible, which could induce selection bias.

\section{Discussion}

The multiple testing procedure by \citet{meinshausen2006false} is a good example of an `exploratory' method \citep{goeman2011multiple}. It offers the researcher freedom to select, based on the data, a set of hypotheses of interest and to obtain a confidence statement on these post hoc selected hypotheses. Until now it was the only permutation-based method that provides simultaneous confidence bounds for the FDP or exceedance control of the FDP.

The methods in this paper allow the user to specify a range of \emph{p}-value thresholds of interest, as well as a set of candidate confidence envelopes. Moreover, the iterative method allows choosing a parameter $s$, which influences power and computational intensity. Various choices for these parameters have been considered in this paper, and future work may provide additional  guidelines for choosing these.

Our methodology relies on an assumption of joint invariance, which underlies most existing permutation-based multiple testing methods. This assumption needs to be argued for in concrete cases, for example as in Section \ref{secdataM}.
\color{black}

In this work we discuss only \emph{p}-values as test statistics, but many of the results can in principle be generalized to arbitrary test statistics  (with possibly unknown null distribution). \color{black} Correspondingly, when \emph{p}-values are used, these are not required to be exact.

\bibliographystyle{myplainnat}
\bibliography{references}

\newpage
\section*{Supplementary material: improved bounds by closed testing}
\label{SM}
\citet{goeman2011multiple} show how closed testing \citep{marcus1976closed} can be used to obtain simultaneous upper bounds for the FDP. 
As will be seen, this result is equivalent to that in \citet{genovese2006exceedance}.
By relating Theorem 1 in our paper to this method, we will  derive a uniform improvement of the envelope $B^{\mathrm{m}}$ of Theorem 1. 

For each nonempty $\I\subseteq \{1,...,m\}$, denote by $H_\I$ the intersection hypothesis $\bigcap_{i\in \I} H_i$. Suppose that for each nonempty $\I\subseteq \{1,...,m\}$ a test for $H_\I$ is defined and suppose $H_\N$ is rejected by its test with probability at most $\alpha$. These $2^m-1$ tests are called \emph{local tests}. The \emph{closed testing procedure} rejects all $H_\I$ for which all $H_\J$ with $\J\supseteq \I$ are rejected.

\citet{genovese2004stochastic, genovese2006exceedance} formulate the FDP  bounds as follows. 
We slightly generalize their setup, since we consider any level-$\alpha$ local tests. 
Let $\U$ be the set of $\B\subseteq\{1,...,m\}$ for which $H_\B$ is not rejected by its local test. For $\K\subseteq \{1,...,m\}$, \citet{genovese2006exceedance} consider the bound
\begin{equation} \label{vgw}
\overline{V}_{ct}(\K)=\max\{ \# \B\cap \K: \B\in \U \},
\end{equation}
where the maximum is defined to be zero if the set is empty.  
The following holds.

\begin{theorem} \label{gw}
Uniformly over all $\K\subseteq \{1,...,m\}$,  $\overline{V}_{\text{ct}}(\K)$ is a $(1-\alpha)$-upper bound for $\# \N\cap \K$, i.e.
$$\mathbb{P} \Bigg [\bigcap_{\K\subseteq \{1,...,m\}} \big \{ \# \N\cap \K \leq \overline{V}_{\text{ct}}(\K) \big \} \Bigg ] \geq 1-\alpha.$$
\end{theorem}
\begin{proof}
With probability at least $1-\alpha$, $H_\N$ is not rejected by its local test, and then $\#  \N\cap \K\leq \overline{V}_{\text{ct}}(\K)$ for all $\K\subseteq\{1,...,m\}$.
\end{proof}
\noindent Note that $\#\N \cap \K$ is the number of false positives if $\K$ is the rejected set. Thus the theorem provides  bounds for the numbers of false positives that are uniform over all possible rejected sets.

It turns out that the bounds $\overline{V}_{ct}(\K)$ are equal to the bounds constructed in \citet{goeman2011multiple}. They consider
$$\mathcal{C}:=\{\I\subseteq\{1,...,m\}: \quad H_\I \text{ is rejected by the closed testing procedure}\}.$$
For each  $\K\subseteq \{1,...,m\}$ they define the bound as
\begin{equation} \label{vct}
\max \{\#\I: \text{ } \I \subseteq {\K},  \text{ } \I\not\in \mathcal{C}\}, 
\end{equation}
Uniformly over all $\K\subseteq \{1,...,m\}$,  \eqref{vct} is a $(1-\alpha)$-upper bound for $\# \N\cap \K$. To prove this, note that with probability at least $1-\alpha$, $H_\N$ is not rejected by its local test, and then $\N \cap \K \not\in \mathcal{C}$ for all $\K\subseteq \{1,...,m\}$.

We now show that the bounds \eqref{vgw}  and  \eqref{vct} are equal, which has never been noted to our knowledge.

\begin{theorem} \label{gw2}
The bounds \eqref{vgw}  and  \eqref{vct} are equal for every $\K\subseteq\{1,...,m\}$.
\end{theorem}
\begin{proof}
We are done if we show that

\begin{align}
  &\max\{\#\B\cap \K: \B\in \mathcal{U}\}= \nonumber \\
   &\max\{\#\B\cap \K: \B\in \mathcal{U} \text{ and } \B\cap \K \not\in\mathcal{C}  \}= \label{1steq} \\
  &\max\{\#\B\cap \K: \B\subseteq \{1,...,m\} \text{ and } \B\cap \K \not\in\mathcal{C}  \}= \label{2ndeq} \\
  &\max\{\#\I: \I\subseteq \K \text{ and } \I \not\in\mathcal{C}  \}. \nonumber
\end{align}

The first and the last equality clearly hold. It is also clear that \eqref{1steq} $\leq$ \eqref{2ndeq}, so it is left to show that \eqref{2ndeq} $\leq$ \eqref{1steq}, which we now do. Note that if $\B\subseteq\{1,...,m\}$ and $\B\cap \K\not \in \mathcal{C}$, then there is a $\B'\in \mathcal{U}$ with $\B'\supseteq \B\cap \K$ and $\B'\cap \K \not\in \mathcal{C}.$ Then obviously 
$\# \B\cap \K \leq  \#\B'\cap \K  \leq $\eqref{1steq}. It follows that  \eqref{2ndeq} $\leq$ \eqref{1steq}.
\end{proof}

The equivalent formulations \eqref{vct} and \eqref{vgw}  are closely related, since in both cases the maximum is taken over all subsets of $\K$ that are not rejected by the closed testing procedure. 
Nevertheless the two formulations suggest different algorithms for computing the upper bound.
If a shortcut exists for the closed testing procedure, then an algorithm based on \eqref{vct} may be faster than one based on \eqref{vgw}.

As an example of a local test, consider the one which rejects $H_\I$ when 
\begin{equation} \label{eq:eventlc3}
\bigcup_{t\in \co} \big\{    R_\I(t)> B_{\I}(t)   \big\},
\end{equation}
where $B_\I$ is  defined in Section 3.1 of our paper. In particular, as noted there, $B_\I$ can be defined as 
\begin{equation} \label{eq:BInonp}
B_\I=\min\Bigg\{B\in  \mathbb{B}: \text{ } w^{-1}\#\Big\{1\leq j \leq w: \bigcap_{t\in \co}\big \{ R_\I^j(t) \leq  B(t) \big \} \Big\} \geq 1-\alpha \Bigg\}.
\end{equation}
Using these local tests in \eqref{vgw} we obtain simultaneous bounds $\overline{V}_{\text{ct}}(\K)$ for all $\K \subseteq \{1,...,m\}$.
Note that the function $B^{\text{ct}}: \co \rightarrow \{1,...,m\}$ given by $B^{\text{ct}}(t) = \overline{V}_{\text{ct}}(\R(t))$ is then a confidence envelope. 
It can be shown that $B^{\text{ct}}(t)\leq B_{\{1,...,m\}}(t)$ for all $t\in \co$, i.e. it is a uniform improvement. (This follows from \citet{goeman2011multiple}, equation (7).) 
If $B_{\I}$ is taken to be \eqref{eq:BInonp}, then $B_{\{1,...,m\}}$ coicides with the envelope $B^{\mathrm{m}}$ of Theorem 1 in our paper, so that  $B^{\text{ct}}$ is a uniform improvement of  $B^{\mathrm{m}}$.

In practice calculation of $\overline{V}_{\text{ct}}(\K)$ is computationally infeasible for large $m$, unless shortcuts are available. This is e.g. the case when the local tests are based on Simes' probability inequality \citep{goeman2016simultaneous}, i.e. when
$B_\I(t)=\#\{1\leq i \leq \#\I:i\alpha/\#\I\leq t\}.$
This parametric method is considered in Section 5 of our paper for comparison with our nonparametric methods.
When $B_{\I}$ is permutation-based, fast exact shortcuts for computing $\overline{V}_{\text{ct}}$  are often not available.

\end{document}